\documentclass[a4paper,UKenglish]{lipics}

\usepackage{microtype}

\def\calI{\mathcal{I}}

\def\bbR{\mathbb{R}}

\theoremstyle{plain}
\newtheorem{observation}[theorem]{Observation}

\makeatletter

\newcommand{\mparagraph}[1]{\vspace{1.0ex \@plus1ex
      \@minus.2ex}\noindent\textbf{#1}\hspace{1em}}

\makeatother

\title{An Improved Algorithm for Diameter-Optimally Augmenting Paths in a Metric Space}
\titlerunning{Diameter-Optimally Augmenting Paths} 

\author{Haitao Wang}
\affil{Department of Computer Science,
Utah State University, Logan, UT 84322, USA.
\texttt{haitao.wang@usu.edu}
}

\authorrunning{H. Wang} 

\Copyright{Haitao Wang}

\keywords{diameter, path graphs, augmenting paths, minimizing diameter, metric space}

\begin{document}

\maketitle

\begin{abstract}
Let $P$ be a path graph of $n$ vertices embedded in a metric
space. We consider the problem of adding a new edge to $P$ such
that the diameter of the resulting graph is minimized.
Previously (in ICALP 2015) the problem was solved in $O(n\log^3 n)$ time.
In this paper, based on new observations and different algorithmic techniques, we present
an $O(n\log n)$ time algorithm.
\end{abstract}

\section{Introduction}
\label{sec:intro}
Let $P$ be a path graph of $n$ vertices embedded in a metric
space. We consider the problem of adding a new edge to $P$ such
that the diameter of the resulting graph is minimized. The problem is
formally defined as follows.

Let $G$ be a graph and each edge has a non-negative length. The
{\em length} of any path of $G$ is the total length of all edges of the
path. For any two vertices $u$ and $v$ of $G$, we use $d_G(u,v)$ to
denote the length of the shortest path from $u$ to $v$ in $G$. The {\em diameter} of $G$ is
defined as $\max_{u,v\in G}{d_G(u,v)}$.  

Let $P$ be a path graph of $n$ vertices $v_1,v_2,\ldots,v_n$ and
there is an edge $e(v_{i-1},v_i)$ connecting $v_{i-1}$ and $v_i$ for each $1\leq i\leq
n-1$. Let $V$ be the vertex set of $P$.
We assume $(V,|\cdot|)$ is a metric space and $|v_iv_j|$ is the
distance of any two vertices $v_i$ and $v_j$ of $V$. Specifically, the
following properties hold: (1) the triangle inequality:
$|v_iv_k|+|v_kv_j|\geq |v_iv_j|$; (2) $|v_iv_j|=|v_jv_i|\geq 0$; (3) $|v_iv_j|=0$ if $i=j$.
In particular, for each edge $e(v_{i-1},v_i)$ of $P$, its length
is equal to $|v_{i-1}v_i|$.
We assume that given any two vertices $v_i$ and $v_j$ of $P$, the
distance $|v_iv_j|$ can be obtained in $O(1)$ time.

Our goal is to find a new edge
$e$ connecting two vertices of $P$ and
add $e$ to $P$, such that the diameter of the resulting graph $P\cup
\{e\}$ is minimized.

The problem has been studied before. Gro{\ss}e et
al.~\cite{ref:GrobeFa15} solved the problem in $O(n\log^3 n)$ time.
In this paper, we present a new algorithm that runs in $O(n\log n)$
time. Our algorithm is based on new observations on the structures of the optimal solution and different algorithmic techniques. Following the previous
work~\cite{ref:GrobeFa15}, we refer to the problem as {\em the
diameter-optimally augmenting path problem}, or DOAP for short.

\subsection{Related Work}
If the path $P$ is in the Euclidean space $\bbR^d$ for a constant $d$, then Gro{\ss}e et
al.~\cite{ref:GrobeFa15} also gave an $O(n+1/\epsilon^3)$ time
algorithm that can find a $(1+\epsilon)$-approximation solution for
the problem DOAP, for any $\epsilon>0$. If $P$ is in the Euclidean plane $\bbR^2$,
De Carulfel et al.~\cite{ref:DeCarufelMi16} gave a linear time algorithm for
adding a new edge to $P$ to minimize the {\em continuous diameter}
(i.e., the diameter is defined with respect to all points of $P$, not only vertices).

The more general problem and many variations have also been studied
before, e.g., see
\cite{ref:AlonDe00,ref:BiloIm12,ref:DemaineMi12,ref:FratiAu15,ref:GaoTh13,ref:IshiiAu13,ref:LiOn92,ref:SchooneDi97} and the references therein.
Consider a general graph $G$ in which edges have non-negative lengths.
For an integer $k$, the goal of the general problem
is to compute a set $F$ of $k$ new edges
and add them to $G$ such that the resulting graph has the minimum
diameter. The problem is NP-hard \cite{ref:SchooneDi97} and some other variants are
even W[2]-hard \cite{ref:FratiAu15,ref:GaoTh13}.
Approximation results have been given for the
general problem and many of its variations, e.g., see
\cite{ref:BiloIm12,ref:FratiAu15,ref:LiOn92}.
The upper bounds and lower bounds on the values of the diameters
of the augmented graphs have also been investigated, e.g. see
\cite{ref:AlonDe00,ref:IshiiAu13}.

Since diameter is an important metric of network performance, which
measures the worst-case cost between any two nodes of the network, as discussed in
\cite{ref:BiloIm12,ref:DemaineMi12}, the problem of
augmenting graphs for minimizing the diameter and its variations have
many practical applications, such as in data networks, telephone networks,
transportation networks, scheduling problems, etc.

As an application of our problem DOAP, consider the following scenario in transportation networks. Suppose there is a highway that connects several cities. In order to reduce the transportation time, we want to build a new highway connecting two cities such that the distance between the farthest two cities using both highways is minimized. Clearly, this is a problem instance of DOAP.

\subsection{Our Approaches}

To tackle the problem, Gro{\ss}e et al.~\cite{ref:GrobeFa15}
first gave an $O(n\log n)$ time algorithm for
the {\em decision version} of the problem:
Given any value $\lambda$, determine whether it is
possible to add a new edge $e$ into $P$ such that the diameter
of the resulting graph is at most $\lambda$. Then, by implementing the
above decision algorithm in a parallel fashion and applying Megiddo's
parametric search \cite{ref:MegiddoAp83}, they solved
the original problem DOAP in $O(n\log^3 n)$ time~\cite{ref:GrobeFa15}. For
differentiation, we referred to the original problem DOAP as the {\em optimization
problem}.

Our improvement over the previous work \cite{ref:GrobeFa15} is twofold.

First, we solve the decision problem in $O(n)$ time. Our algorithm is
based on the $O(n\log n)$ time algorithm in the previous
work~\cite{ref:GrobeFa15}. However, by discovering new observations on
the problem structure and
with the help of the range-minima data structure~\cite{ref:BenderTh00,ref:HarelFa84},
we avoid certain expensive operations
and eventually achieve the $O(n)$ time complexity.

Second, comparing with the decision problem, our algorithm for the
optimization problem is completely different from the previous work
\cite{ref:GrobeFa15}. Let $\lambda^*$ be the diameter of the resulting
graph in an optimal solution.  Instead of using the parametric search,
we identify a set $S$ of candidate values such that $\lambda^*$ is in $S$ and
then we search $\lambda^*$ in $S$ using our algorithm for the decision
problem. However, computational difficulties arise for this approach
due to that the set $S$ is too large ($|S|=\Omega(n^2)$)
and computing certain values of
$S$ is time-consuming (e.g., for certain values of $S$, computing each
of them takes $O(n)$ time). To circumvent these difficulties, our algorithm
has several steps. In each step, we shrink $S$ significantly such
that $\lambda^*$ always remains in $S$. More
importantly, each step will obtain certain formation, based on which
the next step can further reduce $S$. After several steps, the
size of $S$ is reduced to $O(n)$ and all the remaining values of $S$ can be
computed in $O(n\log n)$ time. At this point we can use our decision
algorithm to find $\lambda^*$ from $S$ in additional $O(n\log n)$
time.
Equipped with our linear time algorithm
for the decision problem and utilizing several other algorithmic
techniques such as the sorted-matrix searching
techniques~\cite{ref:FredericksonGe84,ref:FredericksonFi83} and
range-minima data structure~\cite{ref:BenderTh00,ref:HarelFa84}, we
eventually solve the optimization problem in $O(n\log n)$ time.

The rest of the paper is organized as follows. In Section
\ref{sec:pre}, we introduce some notation and observations.
In Section \ref{sec:decision},
we present our algorithm for the decision problem. The optimization problem is solved in Section \ref{sec:optimization}.

\section{Preliminaries}
\label{sec:pre}
In this section, we introduce some notation and observations, some
of which are from Gro{\ss}e et al.~\cite{ref:GrobeFa15}.

For any two vertices $v_i$ and $v_j$ of $P$, we use $e(v_i,v_j)$ to
denote the edge connecting $v_i$ and $v_j$ in the metric space. Hence,
$e(v_i,v_j)$ is in $P$ if and only if $|i-j|=1$.  The length of
$e(v_i,v_j)$ is $|v_iv_j|$.

For any $i$ and $j$ with $1\leq i\leq j\leq n$, we use $G(i,j)$ to
denote the resulting graph by adding the edge $e(v_i,v_j)$ into
$P$. If $i=j$, $G(i,j)$ is essentially $P$. Let $D(i,j)$ denote the diameter of
$G(i,j)$. 

Our goal for the optimization problem DOAP is to find a pair of
indices $(i,j)$ with $1\leq i\leq j\leq n$ such that $D(i,j)$ is
minimized. Let $\lambda^*=\min_{1\leq i\leq j\leq n}D(i,j)$, i.e.,
 $\lambda^*$ is the diameter in an optimal solution.


Given any value $\lambda$, the decision problem is to determine
whether $\lambda\geq \lambda^*$, or in other words, determine whether
there exist a pair $(i,j)$ with $1\leq i\leq j\leq n$ such that
$D(i,j)\leq \lambda$.
If yes, we say that $\lambda$ is a {\em feasible} value.

Recall that for any graph $G$, $d_{G}(u,v)$ refers to the length of the shortest
path between two vertices $u$ and $v$ in $G$.

\begin{figure}[t]
\begin{minipage}[t]{\textwidth}
\begin{center}
\includegraphics[height=1.0in]{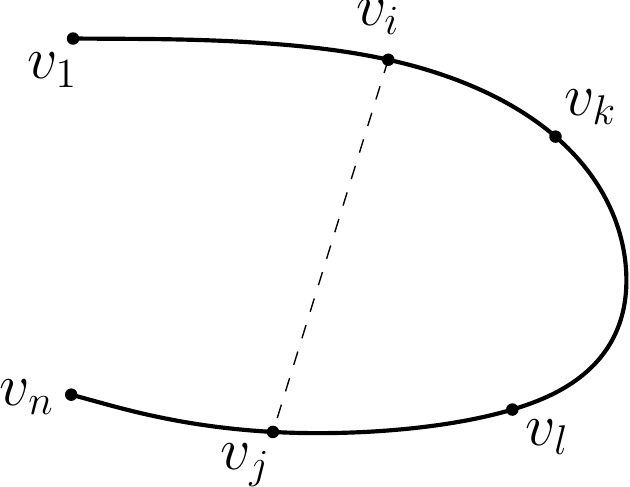}
\caption{\footnotesize Illustrating the resulting graph after a new edge $e(v_i,v_j)$ is added.}
\label{fig:funs}
\end{center}
\end{minipage}
\vspace{-0.15in}
\end{figure}

Consider any pair of indices $(i,j)$ with $1\leq i\leq j\leq n$. We
define $\alpha(i,j)$, $\beta(i,j)$, $\gamma(i,j)$, and $\delta(i,j)$
as follows (refer to Fig.~\ref{fig:funs}).

\begin{definition}
\begin{enumerate}
\item
Define $\alpha(i,j)$ to be the largest shortest path length in
$G(i,j)$ from $v_1$ to all vertices $v_k$ with $k\in [i,j]$, i.e.,
$\alpha(i,j)=\max_{i\leq k\leq j}d_{G(i,j)}(v_1,v_k)$.
\item
Define $\beta(i,j)$ to be the largest shortest path length in
$G(i,j)$ from $v_n$ to all vertices $v_k$ with $k\in [i,j]$, i.e.,
$\beta(i,j)=\max_{i\leq k\leq j}d_{G(i,j)}(v_k,v_n)$.
\item
Define $\gamma(i,j)$ to be the largest shortest path length in
$G(i,j)$ from $v_k$ to $v_l$ for any $k$ and $l$  with $i\leq k\leq
l\leq j$, i.e.,
$\gamma(i,j)=\max_{i\leq k\leq l\leq j}d_{G(i,j)}(v_k,v_l)$.
\item
Define $\delta(i,j)$ to be the shortest path length in
$G(i,j)$ from $v_1$ to $v_n$, i.e.,  $\delta(i,j)=d_{G(i,j)}(v_1,v_n)$.
\end{enumerate}
\end{definition}

It can be verified (also shown in \cite{ref:GrobeFa15}) that the
following observation holds.

\begin{observation}\label{obser:10}{\em(\cite{ref:GrobeFa15})}
$D(i,j)=\max\{\alpha(i,j), \beta(i,j), \gamma(i,j), \delta(i,j)\}$.
\end{observation}


Further, due to the triangle inequality of the metric space, the
following monotonicity properties hold.

\begin{observation}\label{obser:20}{\em(\cite{ref:GrobeFa15})}
\begin{enumerate}
\item
For any $1\leq i\leq j\leq n-1$, $\alpha(i,j)\leq \alpha(i,j+1)$,
$\beta(i,j)\geq \beta(i,j+1)$, $\gamma(i,j)\leq \gamma(i,j+1)$,
and $\delta(i,j)\geq \delta(i,j+1)$.

\item
For any $1\leq i<j\leq n$,
$\alpha(i,j)\leq \alpha(i+1,j)$,
$\beta(i,j)\geq \beta(i+1,j)$, $\gamma(i,j)\geq \gamma(i+1,j)$,
and $\delta(i,j)\leq \delta(i+1,j)$.
\end{enumerate}
\end{observation}

For any pair $(i,j)$ with $1\leq i\leq j\leq n$,
let $P(i,j)$ denote the subpath of $P$ between $v_i$ and $v_j$. Hence,
$d_P(v_i,v_j)$ is the length of $P(i,j)$, i.e.,
$d_P(v_i,v_j)=\sum_{i\leq k\leq j-1}|v_kv_{k+1}|$ if $i<j$ and $d_P(v_i,v_j)=0$ if
$i=j$.

In our algorithms, we will need to compute $f(i,j)$ for each $f\in
\{\alpha,\beta,\gamma,\delta\}$. The next observation gives an algorithm.
The result was also given by Gro{\ss}e et al.~\cite{ref:GrobeFa15} and we present the
proof here for completeness of this paper.

\begin{lemma}\label{lem:preprocess}{\em(\cite{ref:GrobeFa15})}
With $O(n)$ time preprocessing, given any pair $(i,j)$ with $1\leq
i\leq j\leq n$, we can compute $d_P(i,j)$ and $\delta(i,j)$ in $O(1)$
time, and compute $\alpha(i,j)$ and $\beta(i,j)$ in $O(\log n)$ time.
\end{lemma}
\begin{proof}
As preprocessing,
we compute the prefix sum array $A[1\cdots n]$ such
that $A[k]=\sum_{1\leq l\leq k-1}|v_lv_{l+1}|$ for each $k\in [2,n]$ and
$A[1]=0$. This can be done in $O(n)$ time. This finishes our preprocessing.

Consider any pair $(i,j)$ with $1\leq i\leq j\leq n$.
Note that $d_P(v_i,v_j)=A[j]-A[i]$, which can be computed in constant time.

For $\delta(i,j)$, it is easy to see that $\delta(i,j)=\min\{d_P(1,n),d_P(1,i)+|v_iv_j|+d_P(j,n)\}$.
Hence, $\delta(i,j)$ can be computed in constant time.

For $\alpha(i,j)$, if we consider $d_{G(i,j)}(v_1,v_k)$ as a function of $k\in [i,j]$, then $d_{G(i,j)}(v_1,v_k)$ is a unimodal function. More specifically, as $k$ changes from $i$ to $j$, $d_{G(i,j)}(v_1,v_k)$ first increases and then decreases. Hence, $\alpha(i,j)$ can be computed in $O(\log n)$ time by binary search on the
sequence $v_i,v_{i+1},\ldots,v_j$.

Computing $\beta(i,j)$ can be also done in $O(\log n)$ time in a similar way to $\alpha(i,j)$. We omit the details.
\end{proof}

For computing $\gamma(i,j)$, although one may be able to do so in $O(n)$ time, it is not clear to us how to make it in $O(\log n)$ time even with $O(n\log n)$ time preprocessing. As will be seen later, this is the major difficulty for solving the problem DOAP efficiently. We refer to it as the {\em $\gamma$-computation difficulty}. Our main effort will be to circumvent the difficulty by providing alternative and efficient solutions.

For any pair $(i,j)$ with $1\leq i\leq j\leq n$, we use $C(i,j)$ to
denote the cycle $P(i,j)\cup e(v_i,v_j)$.
Consider $d_{G(i,j)}(v_k,v_l)$ for any $k$ and $l$ with $i\leq k\leq l\leq j$. Notice that the shortest path from $v_k$ to $v_l$ in $C(i,j)$ is also a shortest path in $G(i,j)$. Hence, $d_{G(i,j)}(v_k,v_l)=d_{C(i,j)}(v_k,v_l)$. There are two paths in $C(i,j)$ from $v_k$ to $v_l$: one is $P(k,l)$ and the other uses the edge $e(v_i,v_j)$. We use $d^1_{C(i,j)}(v_k,v_l)$ to denote the length of the above second path, i.e., $d^1_{C(i,j)}(v_k,v_l)=d_P(v_i,v_k)+|v_iv_j|+d_P(v_l,v_j)$.
With these notation, we have $d_{C(i,j)}(v_k,v_l)=\min\{d_P(v_k,v_l),d^1_{C(i,j)}(v_k,v_l)\}$. According to the definition of $\gamma(i,j)$, we summarize our discussion in the following observation.

\begin{observation}\label{obser:gamma}
For any pair $(i,j)$ with $1\leq i\leq j\leq n$,
$\gamma(i,j)=\max_{i\leq k\leq l\leq j}d_{C(i,j)}(v_k,v_l)$, with
$d_{C(i,j)}(v_k,v_l)=\min\{d_P(v_k,v_l),d^1_{C(i,j)}(v_k,v_l)\}$ and
$d^1_{C(i,j)}(v_k,v_l)=d_P(v_i,v_k)+|v_iv_j|+d_P(v_l,v_j)$.
\end{observation}

In the following,
to simplify the notation, when the context is clear, we use index $i$
to refer to vertex $v_i$. For example, $d_P(i,j)$ refers to
$d_P(v_i,v_j)$ and $e(i,j)$ refers to $e(v_i,v_j)$.

\section{The Decision Problem}
\label{sec:decision}
In this section, we present our $O(n)$ time algorithm for solving the
decision problem. For any value $\lambda$, our goal is to determine
whether $\lambda$ is feasible, i.e. whether $\lambda\geq \lambda^*$, or equivalently,
whether there is a pair $(i,j)$ with $1\leq i\leq j\leq n$
such that $D(i,j)\leq \lambda$. If yes, our algorithm can also find
such a {\em feasible edge} $e(i,j)$.

By Observation~\ref{obser:10}, $D(i,j)\leq \lambda$ holds if and only if
$f(i,j)\leq \lambda$ for each $f\in \{\alpha,\beta,\gamma,\delta\}$.
To determine whether $\lambda$ is feasible, our algorithm will determine for each
$i\in [1,n]$, whether there exists $j\in [i,n]$ such that
$f(i,j)\leq \lambda$ for each $f\in \{\alpha,\beta,\gamma,\delta\}$.

For any fixed $i\in [1,n]$, we consider $\alpha(i,j)$, $\beta(i,j)$, $\gamma(i,j)$,
and $\delta(i,j)$ as functions of $j\in [i,n]$.
In light of Observation~\ref{obser:20}, $\alpha(i,j)$ and $\gamma(i,j)$ are
monotonically increasing and $\beta(i,j)$ and $\delta(i,j)$ are
monotonically decreasing (e.g., see Fig.~\ref{fig:monotone}).
We define four indices $I_i(f)$ for
$f\in \{\alpha,\beta,\gamma,\delta\}$ as follows. Refer to Fig.~\ref{fig:monotone}.

\begin{figure}[t]
\begin{minipage}[t]{\textwidth}
\begin{center}
\includegraphics[height=1.2in]{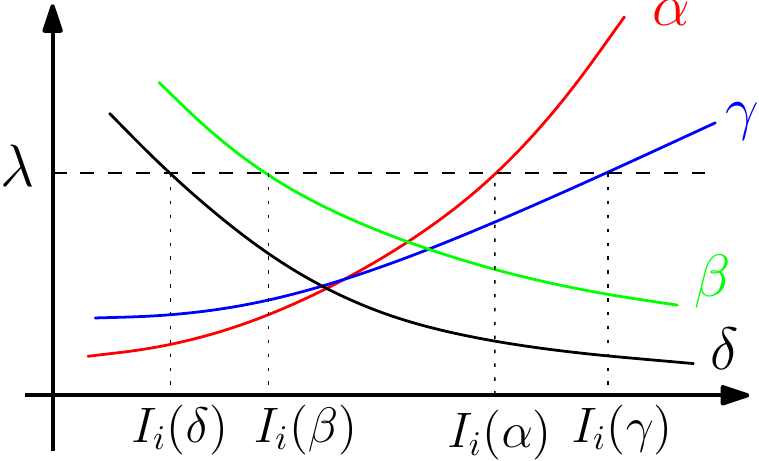}
\caption{\footnotesize Illustrating $f(i,j)$ as $j$ changes in $[i,n]$ and $I_i(f)$ for $f\in \{\alpha,\beta,\gamma,\delta\}$.}
\label{fig:monotone}
\end{center}
\end{minipage}
\end{figure}

\begin{definition}
Define $I_i(\alpha)$ to be the
largest index $j\in [i,n]$ such that $\alpha(i,j)\leq \lambda$.
We define $I_i(\gamma)$ similarly to $I_i(\alpha)$.
If $\beta(i,n)\leq \lambda$, then define $I_i(\beta)$ to be the
smallest index $j\in [i,n]$ such that $\beta(i,j)\leq \lambda$;
otherwise, let $I_i(\beta)=\infty$.  We define $I_i(\delta)$ similarly to $I_i(\beta)$.
\end{definition}


As discussed in \cite{ref:GrobeFa15}, $\lambda$ is feasible if and only
if $[1,I_i(\alpha)]\cap [I_i(\beta),n]\cap [1,I_i(\gamma)]\cap
[I_i(\delta),n]\neq \emptyset$ for some $i\in [1,n]$.
By Observation~\ref{obser:20}, we have the following lemma.

\begin{lemma}\label{lem:10}
For any $i\in [1,n-1]$, $I_i(\alpha)\geq I_{i+1}(\alpha)$,
$I_i(\beta)\geq I_{i+1}(\beta)$,
$I_i(\gamma)\leq I_{i+1}(\gamma)$, and $I_i(\delta)\leq I_{i+1}(\delta)$  (e.g., see Fig.~\ref{fig:shift}).
\end{lemma}
\begin{proof}
According to Observation~\ref{obser:20}, $\alpha(i,j)\leq \alpha(i+1,j)$. This implies that $I_i(\alpha)\geq I_{i+1}(\alpha)$ by the their definitions (e.g., see Fig.~\ref{fig:shift}).
The other three cases for $\beta$, $\gamma$, and $\delta$ are similar.
\end{proof}

\begin{figure}[h]
\begin{minipage}[t]{\textwidth}
\begin{center}
\includegraphics[height=1.2in]{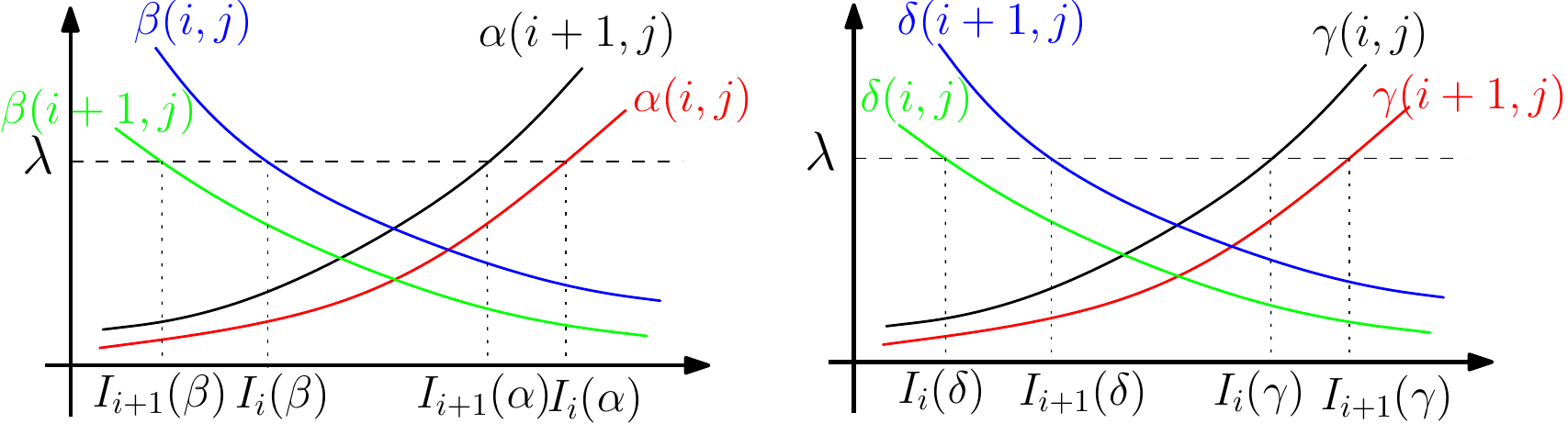}
\caption{\footnotesize Illustrating $f(i,j)$ and $f(i+1,j)$ as $j$ changes and $I_i(f)$ and $I_{i+1}(f)$ for $f\in \{\alpha,\beta,\gamma,\delta\}$.}
\label{fig:shift}
\end{center}
\end{minipage}
\vspace{-0.15in}
\end{figure}

\subsection{Computing $I_i(\alpha)$, $I_i(\beta)$, and $I_i(\delta)$
for all $i\in [1,n]$}
\label{sec:computeI}

In light of Lemma~\ref{lem:10},
for each $f\in \{\alpha,\beta,\delta\}$, we compute $I_i(f)$ for all
$i=1,2,\ldots,n$ in $O(n)$ time, as follows.

We discuss the case for $\delta$ first.
According to Lemma~\ref{lem:preprocess}, $\delta(i,j)$
can be computed in constant time for any pair $(i,j)$ with $1\leq
i\leq n$. We can compute $I_i(\delta)$ for
all $i\in [1,n]$ in $O(n)$ time by the following simple algorithm. We first compute
$I_1(\delta)$, which is done by computing $\delta(1,j)$ from $j=1$
incrementally until the first time $\delta(1,j)\leq \lambda$. Then, to
compute $I_2(\delta)$, we compute $\delta(2,j)$ from $j=I_1(\delta)$
incrementally until the first time $\delta(2,j)\leq \lambda$. Next, we
compute $I_i(\delta)$ for $i=3,4,\ldots,n$ in the same way. The total
time is $O(n)$. The correctness is based on the monotonicity property
of $I_i(\delta)$ in Lemma~\ref{lem:10}.

To compute $I_i(\alpha)$ or $I_i(\beta)$ for $i=1,2,\ldots,n$, using a similar
approach as above, we can only have an $O(n\log n)$ time algorithm
since computing each
$\alpha(i,j)$ or  $\beta(i,j)$  takes $O(\log n)$ time by
Lemma~\ref{lem:preprocess}. In the following Lemma~\ref{lem:20}, we give another approach
that only needs $O(n)$ time.

\begin{lemma}\label{lem:20}
$I_i(\alpha)$ and $I_i(\beta)$ for all $i=1,2,\ldots,n$ can be computed in $O(n)$ time.
\end{lemma}
\begin{proof}
We only discuss the case for $\beta$ since the other case for $\alpha$ is
analogous.

The key idea is that for each pair $(i,j)$, instead of
computing the exact value of $\beta(i,j)$, it is sufficient to
determine whether $\beta(i,j)\leq \lambda$.
In what follows, we show that with $O(n)$ time preprocessing, we
can determine whether $\beta(i,j)\leq \lambda$ in $O(1)$ time for any
index pair $(i,j)$ with $1\leq i\leq j\leq n$.

Let $k$ be the smallest index in $[1,n]$ such that $d_P(k,n)\leq \lambda$.
This implies that $d_P(k-1,n)>\lambda$ if $k>1$.
As preprocessing, we compute the index $k$, which can be easily done in
$O(n)$ time (or even in $O(\log n)$ time by binary search).

Consider any pair $(i,j)$ for $1\leq i\leq j\leq n$.
Our goal is to determine whether $\beta(i,j)\leq
\lambda$.
\begin{enumerate}
\item
If $k\leq i$, then it is vacuously true that
$\beta(i,j)\leq \lambda$.
\item
If $k>j$, then $\beta(i,j)>
\lambda$.
\item
If $i<k\leq j$, a crucial observation is that
$\beta(i,j)\leq \lambda$ if and only if the
length of the path from $v_n$ to $v_{k-1}$ using the new
edge $e(i,j)$, i.e.,
$d_P(i,k-1)+|v_iv_j|+d_P(j,n)$, is less than
or equal to $\lambda$. See Fig.~\ref{fig:computeIbeta}.
Clearly, the above path length can be computed in
constant time, and thus, we can determine whether $\beta(i,j)\leq \lambda$ in
constant time.
\end{enumerate}

\begin{figure}[t]
\begin{minipage}[t]{\textwidth}
\begin{center}
\includegraphics[height=1.2in]{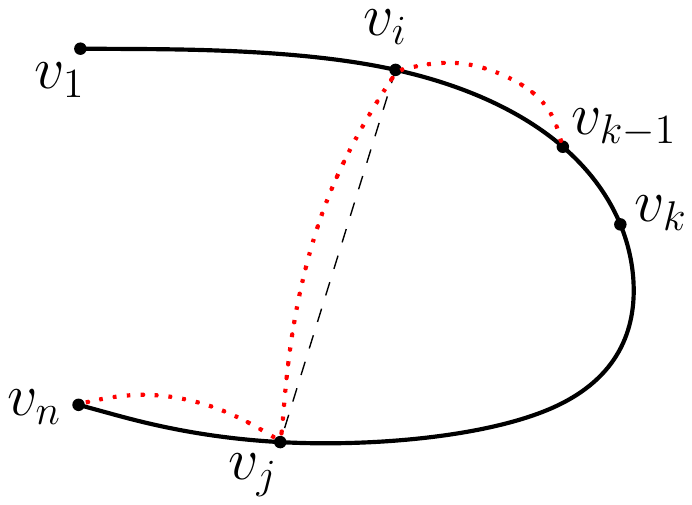}
\caption{\footnotesize Illustrating the path (the dotted curve) from $v_n$ to $v_{k-1}$
using the edge $e(i,j)$.}
\label{fig:computeIbeta}
\end{center}
\end{minipage}
\vspace{-0.15in}
\end{figure}

Therefore, we can determine whether $\beta(i,j)\leq \lambda$ in
constant time for any pair $(i,j)$ with $1\leq i\leq j\leq n$.
With this result, we can use a similar algorithm as the above for
computing $I_i(\delta)$ to compute $I_i(\beta)$ for all $i\in [1,n]$
in $O(n)$ time. The lemma thus follows.
\end{proof}

Due to the $\gamma$-computation difficulty mentioned in Section~\ref{sec:pre},
it is not clear whether it possible to compute $I_i(\gamma)$ for all $i=1,\ldots,n$ in
$O(n\log n)$ time.

Recall that $\lambda$ is feasible if and only if there exists an $i\in
[1,n]$ such that
$[1,I_i(\alpha)]\cap [I_i(\beta),n]\cap [1,I_i(\gamma)]\cap
[I_i(\delta),n]\neq \emptyset$.
Now that $I_i(f)$ for all $i=1,2,\ldots,n$ and
$f\in \{\alpha,\beta,\delta\}$ have been computed but the $I_i(\gamma)$'s are
not known, in the following we will use an ``indirect'' approach to determine
whether the intersection of the above four intervals is empty for
every $i\in [1,n]$.

\subsection{Determining the Feasibility of $\lambda$}

For each $i\in [1,n]$,
define $Q_i=[1,I_i(\alpha)]\cap [I_i(\beta),n]\cap [1,I_i(\gamma)]\cap [I_i(\delta),n]$.
Our goal is to determine whether $Q_i$ is empty for each
$i=1,2,\ldots,n$.

Consider any $i\in [1,n]$. Since $I_i(f)$ for each $f\in
\{\alpha,\beta,\delta\}$ is known, we can determine the intersection
$[1,I_i(\alpha)]\cap [I_i(\beta),n]\cap [I_i(\delta),n]$ in constant
time.
If the intersection is empty, then we know that $Q_i=\emptyset$. In the
following, we assume the intersection is not empty.

Let $a_i$ be the
smallest index in the above intersection. As in \cite{ref:GrobeFa15},
an easy observation is that $Q_i\neq \emptyset$ if and only if $a_i\in
[1,I_i(\gamma)]$. If $a_i\leq i$ (note that $a_i\leq i$ actually implies $a_i=i$ since $a_i\geq I_i(\beta)\geq i$), it is obviously true that $a_i\in
[1,I_i(\gamma)]$ since $i\leq I_i(\gamma)$. Otherwise (i.e., $i<a_i$), according to
the definition of $I_i(\gamma)$, $a_i\in
[1,I_i(\gamma)]$ if and only if $\gamma(i,a_i)\leq \lambda$.
Gro{\ss}e et al.~\cite{ref:GrobeFa15} gave an approach that can
determine whether $\gamma(i,a_i)\leq \lambda$ in $O(\log n)$ time after
$O(n\log n)$ time preprocessing.
In the following, by new observations and with the help of the range
minima data structure \cite{ref:BenderTh00,ref:HarelFa84}, we show that
whether $\gamma(i,a_i)\leq \lambda$  can be determined in constant time after
$O(n)$ time preprocessing.

For each $j\in [1,n]$, define $g_j$ as the largest index $k$ in
$[j,n]$ such that $d_P(j,k)\leq \lambda$. Observe that $g_1\leq
g_2\leq\cdots \leq g_n$.

Consider any $i$ and the corresponding $a_i$ with $i<a_i$.
Our goal is to determine whether $\gamma(i,a_i)\leq \lambda$. Since we
are talking about $\gamma(i,a_i)$, we are essentially considering the
graph $G(i,a_i)$. Recall that $C(i,a_i)$ is the cycle $P(i,a_i)\cup
e(i,a_i)$. By Observation~\ref{obser:gamma},
$\gamma(i,a_i)=\max_{i\leq k\leq l\leq a_i}d_{C(i,a_i)}(k,l)$, and further,
$d_{C(i,j)}(k,l)=\min\{d_P(k,l),d^1_{C(i,a_i)}(k,l)\}$
and $d^1_{C(i,a_i)}(k,l)=d_P(i,k)+|v_iv_{a_i}|+d_P(l,a_i)$.

For any $j\in [i,a_i-1]$, if $g_j\leq a_i-1$, then vertex $g_j+1$ is in
the cycle $C(i,a_i)$. Note that
$d^1_{C(i,a_i)}(j,g_j+1)=d_P(i,j)+|v_iv_{a_i}|+d_P(g_j+1,a_i)$. See
Fig.~\ref{fig:gj}.

\begin{figure}[t]
\begin{minipage}[t]{\textwidth}
\begin{center}
\includegraphics[height=1.0in]{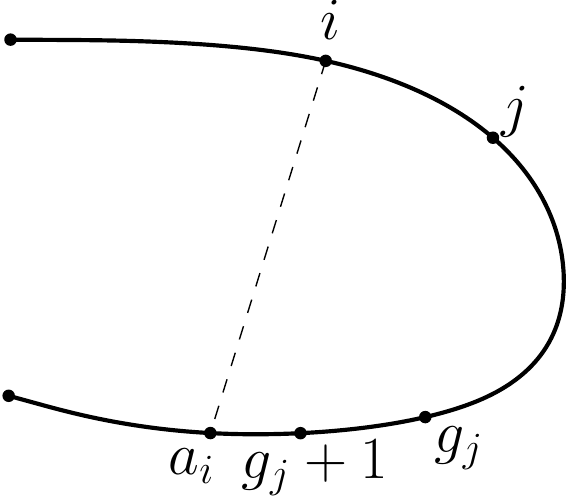}
\caption{\footnotesize Illustrating the graph $G(i,a_i)$ with $g_j+1\leq a_i$.}
\label{fig:gj}
\end{center}
\end{minipage}
\vspace{-0.15in}
\end{figure}

We have the following lemma.
\begin{lemma}\label{lem:310}
$\gamma(i,a_i)\leq \lambda$
if and only if for each $j\in [i,a_i-1]$, either $g_j\geq a_i$ or
$d^1_{C(i,a_i)}(j,g_j+1)\leq \lambda$.
\end{lemma}
\begin{proof}
\begin{itemize}
\item
Suppose $\gamma(i,a_i)\leq \lambda$. Consider any $j\in [i,a_i-1]$
such that $g_j\leq a_i-1$. Below we prove $d^1_{C(i,a_i)}(j,g_j+1)\leq
\lambda$ must hold.

By the definition of $g_j$, it holds that $d_P(j,g_j+1)>\lambda$.
Since  $\gamma(i,a_i)\leq \lambda$ and $d_{C(i,a_i)}(j,g_j+1)\leq
\gamma(i,a_i)$, we obtain $d_{C(i,a_i)}(j,g_j+1)\leq \lambda$. Note that
$d_{C(i,a_i)}(j,g_j+1)=\min\{d_P(j,g_j+1),d^1_{C(i,a_i)}(j,g_j+1)\}$.
Hence, it must hold that $d^1_{C(i,a_i)}(j,g_j+1)\leq \lambda$.

This proves one direction of the lemma.

\item
Suppose it is true that for each $j\in [i,a_i-1]$, either $g_j\geq a_i$ or
$d^1_{C(i,a_i)}(j,g_j+1)\leq \lambda$. We prove $\gamma(i,a_i)\leq \lambda$
below.

Consider any pair of indices $(k,l)$ with $i\leq k\leq l\leq a_i$. To
prove $\gamma(i,a_i)\leq \lambda$, it is sufficient to show that
$d_{C(i,a_i)}(k,l)\leq \lambda$. If $k=l$, then $d_{C(i,a_i)}(k,l)=0$
and thus $d_{C(i,a_i)}(k,l)\leq \lambda$ obviously holds. In the
following we assume $k<l$. This implies that $k\leq a_i-1$. Hence,
either $g_k\geq a_i$ or $d^1_{C(i,a_i)}(k,g_k+1)\leq \lambda$.

Recall that
$d_{C(i,a_i)}(k,l)=\min\{d_P(k,l), d^1_{C(i,a_i)}(k,l)\}$.

If $g_k\geq a_i$, then $l\leq a_i\leq g_k$, and thus, $d_P(k,l)\leq
\lambda$ by the definition of $g_k$. Hence, we obtain $d_{C(i,a_i)}(k,l)\leq \lambda$.

Otherwise, we have $d^1_{C(i,a_i)}(k,g_k+1)\leq \lambda$. If $l\leq g_k$,
we again have $d_P(k,l)\leq\lambda$ and thus $d_{C(i,a_i)}(k,l)\leq
\lambda$. If $l\geq g_k+1$, then $d_P(l,a_i)\leq d_P(g_k+1,a_i)$.
Hence, $d^1_{C(i,a_i)}(k,l)=d_P(i,k)+|v_iv_{a_i}|+d_P(l,a_i)\leq
d_P(i,k)+|v_iv_{a_i}|+d_P(g_k+1,a_i)=d^1_{C(i,a_i)}(k,g_k+1)\leq
\lambda$.
Consequently, we again obtain $d_{C(i,a_i)}(k,l)\leq \lambda$.

This proves the other direction of the lemma.
\end{itemize}
\end{proof}

Recall that $g_1\leq g_2\leq\cdots \leq g_n$. For each $k\in [1,n]$,
define $h_k$ to be the smallest index $j$ in $[1,k]$ with $g_j\geq k$.
Observe that $h_1\leq h_2\leq\cdots\leq h_n$.

Note that if $i<h_{a_i}$, then for each $j\in [i,h_{a_i}-1]$, $g_j<a_i$ and $g_j+1\leq a_i$.
Due to the preceding lemma, we further have the following lemma.
\begin{lemma}\label{lem:320}
$\gamma(i,a_i)\leq \lambda$ if and only if either $h_{a_i}\leq i$ or
$d^1_{C(i,a_i)}(j,g_j+1)\leq \lambda$ holds for each $j\in [i,h_{a_i}-1]$.
\end{lemma}
\begin{proof}
\begin{itemize}
\item
Suppose $\gamma(i,a_i)\leq \lambda$. If $h_{a_i}\leq i$, then we do not need to prove anything. In the following, we assume $h_{a_i}>i$. Consider any $j\in [i,h_{a_i}-1]$. Our goal is to show that $d^1_{C(i,a_i)}(j,g_j+1)\leq \lambda$ holds.

Indeed, since $\gamma(i,a_i)\leq \lambda$, by Lemma~\ref{lem:310}, we have either $g_j\geq a_i$ or $d^1_{C(i,a_i)}(j,g_j+1)\leq \lambda$. Since $j\in [i,h_{a_i}-1]$, $g_j<a_i$. Hence, it must be that $d^1_{C(i,a_i)}(j,g_j+1)\leq \lambda$.

This proves one direction of the lemma.

\item
Suppose either $h_{a_i}\leq i$ or $d^1_{C(i,a_i)}(j,g_j+1)\leq
\lambda$ holds for each $j\in [i,h_{a_i}-1]$. Our goal is to show that
$\gamma(i,a_i)\leq \lambda$. Consider any $k\in [i,a_i-1]$. By
Lemma~\ref{lem:310}, it is sufficient to show that either $g_k\geq a_i$ or
$d^1_{C(i,a_i)}(k,g_k+1)\leq \lambda$.

If $h_{a_i}\leq i$, then since $k\geq i$, we obtain $g_k\geq a_i$ by the definition of $h_{a_i}$.

Otherwise, $d^1_{C(i,a_i)}(j,g_j+1)\leq \lambda$ holds for each $j\in [i,h_{a_i}-1]$. If $k\geq h_{a_i}$, then we still have $g_k\geq a_i$. Otherwise, $k$ is in $[i,h_{a_i}-1]$, and thus it holds that  $d^1_{C(i,a_i)}(k,g_k+1)\leq \lambda$.

This proves the other direction of the lemma.
\end{itemize}
\end{proof}

Let $|C(i,a_i)|$ denote the total length of the cycle $C(i,a_i)$,
i.e., $|C(i,a_i)|=d_P(i,a_i)+|v_iv_{a_i}|$.
The following observation is crucial because it immediately leads to
our algorithm in Lemma~\ref{lem:30}.
\begin{observation}\label{obser:110}
$\gamma(i,a_i)\leq \lambda$
if and only if either $h_{a_i}\leq i$ or $\min_{j\in
[i,h_{a_i}-1]}\{d_P(j,g_j+1)\} \geq |C(i,a_i)|-\lambda$.
\end{observation}
\begin{proof}
Suppose $h_{a_i}> i$. Then, for each $j\in [i,h_{a_i}-1]$, $g_j<a_i$ and $g_j+1\leq a_i$.
Note that $d^1_{C(i,a_i)}(j,g_j+1) = |C(i,a_i)|-d_P(i,g_j+1)$.
Hence, $d^1_{C(i,a_i)}(j,g_j+1)\leq \lambda$ is equivalent
to $d_P(j,g_j+1)\geq |C(i,a_i)|-\lambda$.
Therefore, $d^1_{C(i,a_i)}(j,g_j+1)\leq \lambda$ holds for each $j\in
[i,h_{a_i}-1]$ if and only if $\min_{j\in
[i,h_{a_i}-1]}\{d_P(j,g_j+1)\} \geq |C(i,a_i)|-\lambda$.

By Lemma~\ref{lem:320}, the observation follows.
\end{proof}


\begin{lemma}\label{lem:30}
With $O(n)$ time preprocessing, given any $i\in [1,n]$ and the
corresponding $a_i$ with $i< a_i$, whether $\gamma(i,a_i)\leq \lambda$
can be determined in constant time.
\end{lemma}
\begin{proof}
As preprocessing, we first compute $g_j$ for all $j=1,2,\ldots, n$,
which can be done in $O(n)$ time due to the
monotonicity property $g_1\leq g_2\leq \ldots \leq g_n$. Then, we compute $h_k$ for all
$k=1,2,\ldots,n$, which can also be done in $O(n)$ time due to the
monotonicity property $h_1\leq h_2\leq \ldots \leq h_n$.
Next, we compute an array $B[1,\ldots,n]$ with $B[j]=d_P(j,g_j+1)$
for each $j\in [1,n]$ (let $d_P(j,g_j+1)=\infty$ if $g_j+1>n$).
We build a range-minima data structure on $B$
\cite{ref:BenderTh00,ref:HarelFa84}. The range
minima data structure can be built in $O(n)$ time such that given any
pair $(i,j)$ with
$1\leq i\leq j\leq n$, the minimum value of the subarray $B[i\cdots j]$ can be
returned in constant time \cite{ref:BenderTh00,ref:HarelFa84}.
This finishes the preprocessing step, which takes $O(n)$ time in
total.

Consider any $i$ and the corresponding $a_i$ with $i<a_i$. Our goal is to determine
whether $\gamma(i,a_i)\leq \lambda$, which can be done in $O(1)$ time
as follows.

By Observation~\ref{obser:110}, $\gamma(i,a_i)\leq \lambda$
if and only if either $h_{a_i}\leq i$ or $\min_{j\in
[i,h_{a_i}-1]}\{d_P(j,g_j+1)\} \geq |C(i,a_i)|-\lambda$.
Since $h_{a_i}$ has been computed in the preprocessing, we check whether
$h_{a_i}\leq i$ is true. If yes, then we are done with the assertion
that
$\gamma(i,a_i)\leq \lambda$. Otherwise, we need to determine whether $\min_{j\in
[i,h_{a_i}-1]}\{d_P(j,g_j+1)\} \geq |C(i,a_i)|-\lambda$ holds. To this
end, we first compute $\min_{j\in [i,h_{a_i}-1]}\{d_P(j,g_j+1)\}$ in
constant time by querying the range-minima data structure on $B$ with
$(i,h_{a_i}-1)$. Note that $|C(i,a_i)|$ can be computed in constant time. Therefore, we can determine
whether $\gamma(i,a_i)\leq \lambda$ in $O(1)$ time.
This proves the lemma.
\end{proof}

With Lemma~\ref{lem:30}, the decision problem can be solved in $O(n)$
time.
The proof of the following theorem summarizes our algorithm.

\begin{theorem}\label{theo:decision}
Given any $\lambda$, we can determine whether $\lambda$ is feasible in $O(n)$
time, and further, if $\lambda$ is feasible, a feasible edge can  be found in $O(n)$
time.
\end{theorem}
\begin{proof}
First, we do the preprocessing in Lemma~\ref{lem:preprocess} in
$O(n)$ time. Then, for each $f\in \{\alpha,\beta,\delta\}$, we compute
$I_i(f)$ for all $i=1,2,\ldots,n$, in $O(n)$ time. We also do the
preprocessing in Lemma~\ref{lem:30}.

Next, for each $i\in [1,n]$, we do the following. Compute the
intersection $[1,I_i(\alpha)]\cap [I_i(\beta),n]\cap [I_i(\delta),n]$
in constant time. If the intersection is empty, then we are done for
this $i$. Otherwise, obtain the
smallest index $a_i$ in the above intersection. If $a_i\leq i$, then
we stop the algorithm with the assertion that $\lambda$ is feasible and report $e(i,a_i)$ as a feasible edge.
Otherwise, we use
Lemma~\ref{lem:30} to determine whether $\gamma(i,a_i)\leq \lambda$ in
constant time. If yes, we stop the algorithm with the assertion that $\lambda$ is
feasible and report $e(i,a_i)$ as a feasible edge. Otherwise, we proceed on $i+1$.

If the algorithm does not stop after we check all $i\in [1,n]$, then
we stop the algorithm with the assertion that $\lambda$ is not
feasible. Clearly, we spend $O(1)$ time on each $i$, and thus, the
total time of the algorithm is $O(n)$.
\end{proof}

\section{The Optimization Problem}
\label{sec:optimization}

In this section, we present our algorithm that solves the optimization
problem in $O(n\log n)$ time, by making use of our algorithm for the
decision problem given in Section \ref{sec:decision} (we will refer to it
as the {\em decision algorithm}).
It is sufficient to compute $\lambda^*$, after which we can use our
decision algorithm to find an optimal new edge in additional $O(n)$
time.

We start with an easy observation that $\lambda^*$ must be equal to
the diameter $D(i,j)$ of $G(i,j)$ for some pair $(i,j)$ with $1\leq
i\leq j\leq n$. Further, by Observation~\ref{obser:10},
$\lambda^*$ is equal to $f(i,j)$ for some $f\in
\{\alpha,\beta,\gamma,\delta\}$ and some pair $(i,j)$ with $1\leq
i\leq j\leq n$.

For each $f\in \{\alpha,\beta,\gamma,\delta\}$, define
$S_f=\{f(i,j)\ |\ 1\leq i\leq j\leq n\}$.
Let $S=\cup_{f\in \{\alpha,\beta,\gamma,\delta\}}S_f$.
According to our discussion above, $\lambda^*$  is in $S$. Further, note
that $\lambda^*$ is the smallest feasible value of $S$.
We will not compute the entire set $S$ since
$|S|=\Omega(n^2)$.
For each $f\in \{\alpha,\beta,\gamma,\delta\}$,
let $\lambda_{f}$ be the smallest feasible value in $S_f$.
Hence, we have $\lambda^*=\min\{\lambda_{\alpha},\lambda_{\beta},
\lambda_{\gamma},\lambda_{\delta}\}$.

In the following, we first compute $\lambda_{\alpha},\lambda_{\beta},\lambda_{\delta}$
in $O(n\log n)$ time by using our decision algorithm and
the sorted-matrix searching techniques \cite{ref:FredericksonGe84,ref:FredericksonFi83}.

\subsection{Computing $\lambda_{\alpha},\lambda_{\beta}$, and $\lambda_{\delta}$}

For convenience, we begin with computing $\lambda_{\beta}$.

We define an $n\times n$ matrix $M[1\cdots n;1\cdots n]$: For each
$1\leq i\leq n$ and $1\leq j\leq n$, define $M[i,j]=\beta(i,j)$ if
$j\geq i$ and $M[i,j]=\beta(i,i)$ otherwise. By
Observation~\ref{obser:20}, the following lemma
shows that $M$ is a sorted matrix in the sense that each row is sorted in descending order from left to right and each column is sorted in descending order from top to bottom.

\begin{lemma}
For each $1\leq i\leq n$, $M[i,j]\geq M[i,j+1]$ for any $j\in [1,n-1]$;
for each $1\leq j\leq n$, $M[i,j]\geq M[i+1,j]$ for any $i\in [1,n-1]$.
\end{lemma}
\begin{proof}
Consider any two adjacent matrix elements $M[i,j]$ and $M[i,j+1]$ in
the same row. If $j\geq
i$, then $M[i,j]=\beta(i,j)$ and $M[i,j+1]=\beta(i,j+1)$.
By Observation~\ref{obser:20}, $M[i,j]\geq M[i,j+1]$.
If $j<i$, then $M[i,j]=M[i,j+1]=\beta(i,i)$.
Hence, in either case, $M[i,j]\geq M[i,j+1]$ holds.

Consider any two adjacent matrix elements $M[i,j]$ and $M[i+1,j]$ in
the same column. If $j\geq i+1$, then $M[i,j]=\beta(i,j)$ and $M[i+1,j]=\beta(i+1,j)$.
By Observation~\ref{obser:20}, we obtain $M[i,j]\geq M[i,j+1]$.
If $j<i+1$, then $M[i,j]=\beta(j,j)$ and
$M[i,j+1]=\beta(j+1,j+1)$. Note that $\beta(j,j)$ is essentially equal
to $d_P(j,n)$ and $\beta(j+1,j+1)$ is equal
to $d_P(j+1,n)$. Clearly, $d_P(j,n)\geq d_P(j+1,n)$.
Hence, in either case, $M[i,j]\geq M[i,j+1]$.
\end{proof}

Note that each element of $S_{\beta}$ is in $M$ and vice versa.
Since $\lambda_{\beta}$ is the smallest feasible value of $S_{\beta}$,
$\lambda_{\beta}$ is also the smallest
feasible value of $M$. We do not construct $M$ explicitly.
Rather, given any $i$ and $j$, we can ``evaluate'' $M[i,j]$ in $O(\log n)$
time since $\beta(i,j)$ can be computed in $O(\log n)$ time if $i\leq
j$ by Lemma~\ref{lem:preprocess}. Using the sorted-matrix searching techniques
\cite{ref:FredericksonGe84,ref:FredericksonFi83}, we can find
$\lambda_{\beta}$ in $M$ by calling our decision algorithm $O(\log n)$
times and evaluating $O(n)$ elements of $M$. The total time on calling
the decision algorithm is $O(n\log n)$ and the total time on
evaluating matrix elements is also $O(n\log n)$. Hence, we can compute
$\lambda_{\beta}$ in $O(n\log n)$ time.

Computing $\lambda_{\alpha}$ and $\lambda_{\delta}$ can done similarly
in $O(n\log n)$ time, although the corresponding sorted matrices may
be defined slightly differently. We omit the details.
However, we cannot compute $\lambda_{\gamma}$ in $O(n\log n)$ time in
the above way, and again this is due to the $\lambda$-computation
difficulty mentioned in Section~\ref{sec:pre}.

Note that having
$\lambda_{\alpha},\lambda_{\beta}$, and $\lambda_{\delta}$ essentially reduces our search
space for $\lambda^*$ from $S$ to
$S_{\gamma}\cup \{\lambda_{\alpha},\lambda_{\beta},\lambda_{\delta}\}$.

We compute $\lambda_1=\min\{\lambda_{\alpha},\lambda_{\beta},\lambda_{\delta}\}$.
Thus, $\lambda^*=\min\{\lambda_1,\lambda_{\gamma}\}$.
Hence, if $\lambda_{\gamma}\geq \lambda_1$, then $\lambda^*=
\lambda_1$ and we are done for computing $\lambda^*$. Otherwise (i.e.,
$\lambda_{\gamma}<\lambda_1$), it must be that $\lambda^*=\lambda_{\gamma}$ and
we need to compute $\lambda_{\gamma}$. To compute $\lambda_{\gamma}$,
again we cannot use the similar way as the above for computing
$\lambda_{\beta}$. Instead, we use the following approach.
We should point out that the success of the approach relies on the
information implied by $\lambda_{\gamma}<\lambda_1$.

\subsection{Computing $\lambda^*$ in the Case $\lambda_{\gamma}<\lambda_1$}

We assume $\lambda_{\gamma}<\lambda_1$. Hence, $\lambda^*=\lambda_{\gamma}$.
Let $e(i^*,j^*)$ be the new edge added to $P$ in an optimal
solution. We also call $e(i^*,j^*)$ an {\em optimal edge}.

Since $\lambda^*=\lambda_{\gamma}<\lambda_1$, we have the following observation.

\begin{observation}\label{obser:40}
If $\lambda_{\gamma}<\lambda_1$ and $e(i^*,j^*)$ is an optimal edge, then
$\lambda^*=\gamma(i^*,j^*)$.
\end{observation}
\begin{proof}
Assume to the contrary that $\lambda^*\neq \gamma(i^*,j^*)$. Then, by
Observation~\ref{obser:10}, $\lambda^*$ is equal to one of
$\alpha(i^*,j^*),\beta(i^*,j^*)$, and $\delta(i^*,j^*)$. Without loss
of generality, assume $\lambda^*=\alpha(i^*,j^*)$. Since
$\alpha(i^*,j^*)$ is in $S_{\alpha}$, $\lambda^*$ must be the smallest
feasible value of $S_{\alpha}$, i.e., $\lambda^*=\lambda_{\alpha}$.
However, this contradicts with that
$\lambda^*=\lambda_{\gamma}<\lambda_1=\min\{\lambda_{\alpha},\lambda_{\beta},\lambda_{\gamma}\}\leq
\lambda_{\alpha}$.
\end{proof}

For any $i\in [1,n]$, for each $f\in \{\alpha,\beta,\gamma,\delta\}$, with respect to $\lambda_1$, we define $I'_i(f)$ in a similar way to $I_i(f)$ defined in Section~\ref{sec:decision} with respect to
$\lambda$ except that we change ``$\leq \lambda$'' to
``$<\lambda_1$''. Specifically, define $I_i'(\alpha)$ to be the largest index $j\in [i,n]$ such that $\alpha(i,j)<\lambda_1$. $I'(\gamma)$ is defined similarly to $I_i'(\alpha)$.
If $\beta(i,n)<\lambda_1$, then define $I_i'(\beta)$ to be the
smallest index $j\in [i,n]$ such that $\beta(i,j)<\lambda_1$;
otherwise $I'_i(\beta)=\infty$.
$I'_i(\delta)$ is defined similarly to $I_i'(\beta)$.
Note that similar monotonicity properties for $I_i'(f)$ with $f\in
\{\alpha,\beta,\gamma,\delta\}$ to those in Lemma~\ref{lem:10} also
hold.

Recall that $e(i^*,j^*)$ is an optimal edge.
An easy observation is that since $\lambda_1$ is strictly larger than $\lambda^*$, the
intersection $[1,I'_{i^*}(\alpha)]\cap [I'_{i^*}(\beta),n]\cap [I'_{i^*}(\delta),n]$
cannot be empty. Let $a_{i^*}$ be the smallest index in the above
intersection. Note that $i^*\leq a_{i^*}$ since $i^*\leq
I'_{i^*}(\beta)\leq a_{i^*}$. The following lemma shows that
$e(i^*,a_{i^*})$ is actually an optimal edge.

\begin{lemma}\label{lem:50}
If $\lambda_{\gamma}<\lambda_1$ and $e(i^*,j^*)$ is an optimal edge, then $j^*=a_{i^*}$.
\end{lemma}
\begin{proof}
For any pair $(i,j)$ with $1\leq i\leq j\leq n$, let $\eta(i,j)=\max\{\alpha(i,j),\beta(i,j),\delta(i,j)\}$.
By Observation~\ref{obser:10}, $D(i,j)=\max\{\gamma(i,j),\eta(i,j)\}$.

We first prove the following {\em claim}: If $\gamma(i^*,a_{i^*})\geq \eta(i^*,a_{i^*})$,
then $j^*=a_{i^*}$ (e.g., see Fig.~\ref{fig:optedge}).

\begin{figure}[t]
\begin{minipage}[t]{\textwidth}
\begin{center}
\includegraphics[height=1.2in]{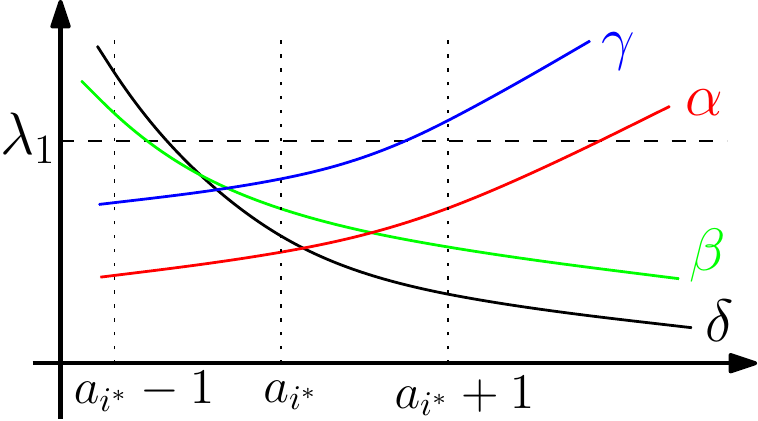}
\caption{\footnotesize Illustrating $f(i^*,j)$ as $j$ changes for
$f\in \{\alpha,\beta,\gamma,\delta\}$. The three indices $a_{i^*}-1$,
$a_{i^*}$, and $a_{i^*}+1$ are shown. }
\label{fig:optedge}
\end{center}
\end{minipage}
\vspace{-0.15in}
\end{figure}

On the one hand, consider any $j\in [i^*,a_{i^*}-1]$. By the definition of $a_{i^*}$,
$\eta(i^*,j)\geq \lambda_1$. Since $\lambda_1>\lambda_{\gamma}=\lambda^*$, $\eta(i^*,j)>\lambda^*$.
By Observation~\ref{obser:10}, $D(i^*,j)\geq \eta(i^*,j)>\lambda^*$.
Hence, $j$ cannot be $j^*$ since otherwise $D(i^*,j)$ would be equal
to $\lambda^*$, incurring contradiction.

On the other hand, consider any $j\in [a^*+1,n]$. By
Observation~\ref{obser:10}, $D(i^*,j)\geq \gamma(i^*,j)$.  By
Observation~\ref{obser:20}, $\gamma(i^*,j)\geq \gamma(i^*,a_{i^*})$. Hence, $D(i^*,j)\geq \gamma(i^*,a_{i^*})$. Further, since
$\gamma(i^*,a_{i^*})\geq \eta(i^*,a_{i^*})$ (the claim hypothesis), we have
$D(i^*,a_{i^*})=\max\{\gamma(i^*,a_{i^*}),\eta(i^*,a_{i^*})\}=\gamma(i^*,a_{i^*})$. Therefore, we obtain $D(i^*,a_{i^*})\leq D(i^*,j)$. This implies that $j^*=a_{i^*}$. Hence, the claim follows.

We proceed to prove the lemma. Based on the above claim, it is
sufficient to show that $\gamma(i^*,a_{i^*})\geq \eta(i^*,a_{i^*})$,
as follows.

Assume to the contrary that $\gamma(i^*,a_{i^*}) < \eta(i^*,a_{i^*})$.
Then, $D(i^*,a_{i^*})=\eta(i^*,a_{i^*})$. According to the definition
of $a_{i^*}$, $\eta(i^*,a_{i^*})<\lambda_1$. Hence,
$D(i^*,a_{i^*})<\lambda_1$. Let $\lambda'=D(i^*,a_{i^*})$. Since
$\lambda'=\eta(i^*,a_{i^*})$, $\lambda'$ is a value in $S_{\alpha}\cup
S_{\beta}\cup S_{\delta}$. Since $\lambda'=D(i^*,a_{i^*})$, $\lambda'$
is a feasible value (i.e., $\lambda'\geq \lambda^*$). Recall that $\lambda_1$ is the smallest feasible value of
$S_{\alpha}\cup S_{\beta}\cup S_{\delta}$. Thus, we obtain
contradiction since $\lambda'<\lambda_1$.

Therefore, $\gamma(i^*,a_{i^*})\geq \eta(i^*,a_{i^*})$ holds. The lemma thus follows.
\end{proof}

Lemma~\ref{lem:50} is crucial because it immediately suggests the following algorithm.

We first compute
the indices $I'_{i}(\alpha), I'_{i}(\beta), I'_{i}(\delta)$ for
$i=1,\ldots, n$. This can be done in $O(n)$ time using the similar
algorithms as those for computing $I_i(\alpha),I_i(\beta),I_i(\delta)$
in Section~\ref{sec:computeI}. In fact, here we can even afford $O(n\log n)$ time to compute these indices. Hence, for simplicity, we can use the similar algorithm as that for computing $I_i(\delta)$ in Section~\ref{sec:computeI} instead of the one in Lemma~\ref{lem:20}. The total time is $O(n\log n)$.

Next, for each $i\in [1,n]$, if $[1,I'_{i}(\alpha)]\cap
[I'_{i}(\beta),n]\cap [I'_{i}(\delta),n]\neq \emptyset$, then we
compute $a_i$, i.e., the smallest index in the above intersection.
Let $\calI$ be the set of index $i$ such that the above
interval  intersection for $i$ is not empty. Lemma~\ref{lem:50} leads to the following observation.

\begin{observation}\label{obser:30}
If $\lambda_{\gamma}<\lambda_1$, then $\lambda^*$ is the smallest feasible value of the set $\{\gamma(i,a_i)\ |\ i\in \calI\}$.
\end{observation}
\begin{proof}
By Lemma~\ref{lem:50}, one of the edges of $\{e(i,a_i)\ |\ i\in \calI\}$ is an optimal edge. By Observation~\ref{obser:40}, $\lambda^*$ is in $\{\gamma(i,a_i)\ |\ i\in \calI\}$. Thus, $\lambda^*$ is the smallest feasible value in $\{\gamma(i,a_i)\ |\ i\in \calI\}$.
\end{proof}

We can further obtain the following ``stronger'' result, although Observation~\ref{obser:30} is sufficient for our algorithm.

\begin{lemma}
If $\lambda_{\gamma}<\lambda_1$, then $\lambda^*=\min_{i\in \calI}\gamma(i,a_i)$.
\end{lemma}
\begin{proof}
For any pair $(i,j)$ with $1\leq i\leq j\leq n$, let $\eta(i,j)=\max\{\alpha(i,j),\beta(i,j),\delta(i,j)\}$.
By Observation~\ref{obser:10}, $D(i,j)=\max\{\gamma(i,j),\eta(i,j)\}$.

We first prove the following {\em claim:} For any $i\in \calI$, $\eta(i,a_i)<\gamma(i,a_i)$. Indeed, assume to the contrary that $\eta(i,a_i)\geq \gamma(i,a_i)$ for some $i\in \calI$. Then, $D(i,a_i)=\eta(i,a_i)$. By the definition of $a_i$, $\eta(i,a_i)<\lambda_1$. Hence, $D(i,a_i)<\lambda_1$. Let $\lambda'=D(i,a_i)$. Note that $\lambda'$ is a feasible value that is in $S_{\alpha}\cup S_{\beta}\cup S_{\delta}$. However, $\lambda'<\lambda_1$ contradicts with that $\lambda_1$ is the smallest feasible value in $S_{\alpha}\cup S_{\beta}\cup S_{\delta}$.

Next, we prove the lemma by using the above claim. For each $i\in \calI$, by the above clam, $D(i,a_i)=\gamma(i,a_i)$, and thus, $\gamma(i,a_i)$ is a feasible value.
By Lemma~\ref{lem:50}, we know that $\lambda^*$ is in $\{\gamma(i,a_i)\ |\ i\in \calI\}$. Therefore, $\lambda^*$ is the smallest value in $\{\gamma(i,a_i)\ |\ i\in \calI\}$. The lemma thus follows.
\end{proof}

Observation~\ref{obser:30} essentially reduces the search space for $\lambda^*$ to $\{\gamma(i,a_i)\ |\ i\in \calI\}$, which has at most $O(n)$ values. It is tempting to first explicitly compute the set and then find $\lambda^*$ from the set. However, again, due to the $\gamma$-computation difficulty, we are not able to compute the set in $O(n\log n)$ time. Alternatively, we use the following approach  to compute $\lambda^*$.

\subsection{Finding $\lambda^*$ in the Set $\{\gamma(i,a_i)\ |\ i\in \calI\}$}
\label{sec:feasibility}

Recall that according to Observation~\ref{obser:gamma}, $\gamma(i,j)=\max_{i\leq k\leq l\leq j}d_{C(i,j)}(k,l)$, with
$d_{C(i,j)}(k,l)=\min\{d_P(k,l),d^1_{C(i,j)}(k,l)\}$ and
$d^1_{C(i,j)}(k,l)=d_P(i,k)+|v_iv_j|+d_P(l,j)$.
Hence, $\gamma(i,j)$ is equal to $d_P(k,l)$ or $d^1_{C(i,j)}(k,l)$ for some $k\leq l$. Therefore, by Observation~\ref{obser:30}, there exists $i\in \calI$ such that $\lambda^*$ is equal to $d_P(k,l)$ or $d^1_{C(i,j)}(k,l)$ for some $k$ and $l$ with $i\leq k\leq l\leq a_i$.

Let $S_p=\{d_P(k,l)\ |\ 1\leq k\leq l\leq n\}$ and $S_c=\{d^1_{C(i,j)}(k,l)\ |\ i\leq k\leq l\leq a_i, i\in \calI\}$. Based on our above discussion, $\lambda^*$ is in $S_p\cup S_c$. Further, $\lambda^*$ is the smallest feasible value in $S_p\cup S_c$.

Let $\lambda_p$ be the smallest feasible value of $S_p$ and let $\lambda_c$ be the smallest feasible value of $S_c$. Hence, $\lambda^*=\min\{\lambda_p, \lambda_c\}$. By using the technique of searching sorted-matrices \cite{ref:FredericksonGe84,ref:FredericksonFi83}, the following lemma computes $\lambda_p$ in $O(n\log n)$ time.

\begin{lemma}\label{lem:70}
$\lambda_p$ can be computed in $O(n\log n)$ time.
\end{lemma}
\begin{proof}
We define an $n\times n$ matrix $M[1\cdots n;1\cdots n]$: For each
$1\leq i\leq n$ and $1\leq j\leq n$, define $M[i,j]=d_P(i,j)$ if
$j\geq i$ and $M[i,j]=0$ otherwise. It is easy to verify that each row of $M$ is sorted in ascending order from the left to right and each column is sorted in ascending order from bottom to top. Consequently, by using the sorted-matrix searching technique \cite{ref:FredericksonGe84,ref:FredericksonFi83}, $\lambda_p$ can be found by calling our decision algorithm $O(\log n)$ times and evaluating $O(n)$ elements of $M$. Clearly, given any $i$ and $j$, we can evaluate $M[i,j]$ in constant time. Hence, $\lambda_p$ can be computed in $O(n\log n)$ time.
\end{proof}

Recall that $\lambda^*=\min\{\lambda_p, \lambda_c\}$.
In the case $\lambda_p\leq \lambda_c$, $\lambda^*=\lambda_p$ and
we are done with computing $\lambda^*$.
In the following, we assume
$\lambda_p> \lambda_c$. Thus, $\lambda^*=\lambda_c$.
With the help of the information implied by $\lambda_p> \lambda_c$,
we will compute $\lambda^*$ in $O(n\log n)$ time.
The details are given below.


For any $j\in [1,n]$, let $g'_j$ denote the largest index $k\in [j,n]$
such that the subpath length $d_P(j,k)$ is {\em strictly smaller} than $\lambda_p$. Note that
the definition of $g_j'$ is similar to $g_j$ defined in
Section~\ref{sec:feasibility}  except that we change ``$\leq \lambda$'' to
``$<\lambda_p$''.

For each $k\in [1,n]$, let $h'_k$
denote the smallest index $j\in [1,k]$ with $g'_j\geq k$. Let $\calI'$ be
the subset of $i\in \calI'$ such that $i\leq h'_{a_i}-1$. Hence, for each $i\in
\calI'$ and each $j\in [i,h'_{a_i}-1]$, $g'_j< a_i$ and thus $g'_j+1\leq
a_i$.

For each $i\in \calI'$, define
$d^1_{\max}(i,a_i)=\max_{j\in [i,h'_{a_i}-1]}d^1_{C(i,j)}(j,g'_j+1)$. The
following lemma gives a way to determine $\lambda^*$.

\begin{lemma}\label{lem:80}
If $\lambda_{\gamma}<\lambda_1$ and $\lambda_c<\lambda_p$,
then $\lambda^*=d^1_{\max}(i,a_i)$ for some
$i\in \calI'$.
\end{lemma}
\begin{proof}
Since $\lambda_{\gamma}<\lambda_1$ and $\lambda_c<\lambda_p$, by our
above discussions, $\lambda^*=\lambda_c$.

By Observation~\ref{obser:30}, $\lambda^*$ is the diameter of the
graph $G(i,a_i)$ for some $i\in \calI$ and $\lambda^*$ is equal to the length of
the shortest path of two vertices $v_k$ and $v_l$ in $C(i,a_i)$ for
$i\leq k\leq l\leq a_i$, i.e.,
$\lambda^*=d_{C(i,a_i)}(k,l)=\min\{d_P(k,l),d^1_{C(i,j)}(k,l)\}$. See
Fig.~\ref{fig:kl}.
Since $\lambda^*=\lambda_c$, we further know that
$\lambda^*=d^1_{C(i,j)}(k,l)$ and $d^1_{C(i,j)}(k,l)\leq d_P(k,l)$.
In fact, $d^1_{C(i,j)}(k,l)< d_P(k,l)$, since otherwise if
$d^1_{C(i,j)}(k,l)= d_P(k,l)$, then $\lambda^* = d_P(k,l)$ would be in
the set $S_p$, contradicting with that $\lambda_p$ is the smallest
feasible value in $S_p$ and $\lambda^*<\lambda_p$.

\begin{figure}[t]
\begin{minipage}[t]{\textwidth}
\begin{center}
\includegraphics[height=1.0in]{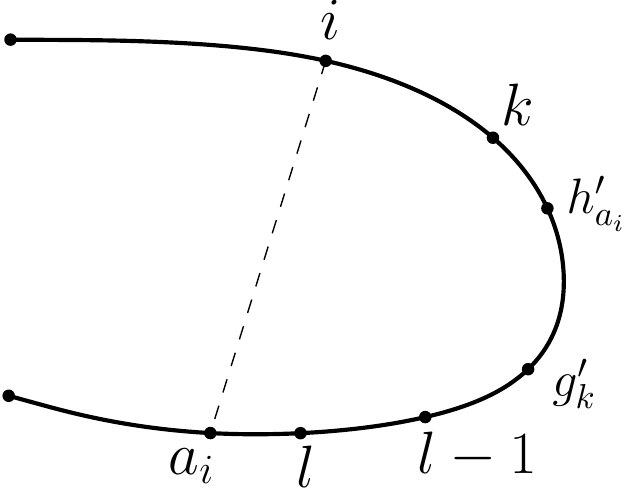}
\caption{\footnotesize Illustrating the graph $G(i,a_i)$ whose
diameter is $\lambda^*$ and $\lambda^*=d_{C(i,a_i)}(k,l)$.}
\label{fig:kl}
\end{center}
\end{minipage}
\vspace{-0.15in}
\end{figure}

For simplicity of discussion, we assume $|v_lv_{l-1}|>0$ (since otherwise
we can keep updating $l$ to $l-1$ until we find $|v_lv_{l-1}|>0$; note
that such an $l$ will eventually be found before we reach $k$ since
$0\leq\lambda^*=d^1_{C(i,j)}(k,l)< d_P(k,l)$).

We prove the following {\em claim:} $d_P(k,l-1)< \lambda_p \leq d_P(k,l)$.

\begin{itemize}
\item
On the one hand, since $\lambda^*=d^1_{C(i,a_i)}(k,l)$ and $|v_lv_{l-1}|>0$, we
obtain that $\lambda^*<d^1_{C(i,a_i)}(k,l-1)$. Since $\lambda^*$ is the
diameter in the graph $G(i,a_i)$,
$d_{G(i,a_i)}(k,l-1)=d_{C(i,a_i)}(k,l-1)\leq \lambda^*$. Further, because
$d_{C(i,a_i)}(k,l-1)=\min\{d_P(k,l-1),d^1_{C(i,a_i)}(k,l-1)\}$ and
$\lambda^*<d^1_{C(i,a_i)}(k,l-1)$, we obtain $d_P(k,l-1)\leq \lambda^*$.
As $\lambda^*=\lambda_c<\lambda_p$, it follows that
$d_P(k,l-1)<\lambda_p$.

\item
On the other hand, assume to the contrary that $\lambda_p>d_P(k,l)$.
Then, since $d_P(k,l)> d^1_{C(i,a_i)}(k,l)=\lambda^*$, $d_P(k,l)$ is a
feasible value. Clearly, $d_P(k,l)$ is in the set $S_p$. However,
$\lambda_p>d_P(k,l)$ contradicts with that $\lambda_p$ is the smallest
feasible value in $S_p$.
\end{itemize}

This proves the claim.
With the claim, we show below
that $\lambda^*=d^1_{\max}(i,a_i)$, which will prove the lemma.

We first show that $i$ is in
$\calI'$, i.e., $i\leq h'_{a_i}-1$. Indeed, since $\lambda_p\leq
d_P(k,l)$ (by the claim),
based on the definition of $g'_k$, it holds that $g'_k<l$ (e.g., see
Fig.~\ref{fig:kl}). Since $l\leq
a_i$, we obtain $g'_k\leq a_i-1$. This implies that $k< h'_{a_i}$ and
thus $k\leq h'_{a_i}-1$. Since $i\leq
k$, $i\leq h'_{a_i}-1$.

It remains to prove $\lambda^*=d^1_{\max}(i,a_i)$.
Indeed, recall that $\lambda^*=d^1_{C(i,a_i)}(k,l)$. Note that the above claim in fact implies that $g_k'=l-1$, and thus, $g_k'+1=l$. Hence, we have $\lambda^*=d^1_{C(i,a_i)}(k,l)= d^1_{C(i,a_i)}(k,g'_k+1)$. Note that $k$ is in $[i,h'_{a_i}-1]$.
Consider any $j\in [i,h'_{a_i}-1]$. To prove $\lambda^*=d^1_{\max}(i,a_i)$,
it is now sufficient to prove $\lambda^*\geq d^1_{C(i,a_i)}(j,g'_j+1)$, as follows.

Recall that $g'_j+1=l\leq a_i$.
Since $\lambda^*$ is the diameter of $G(i,a_i)$,
$d_{G(i,a_i)}(j,g'_j+1)=d_{C(i,a_i)}(j,g'_j+1)\leq \lambda^*$. Recall that
$d_{C(i,a_i)}(j,g'_j+1)=\min\{d_P(j,g'_j+1),d^1_{C(i,a_i)}(j,g'_j+1)\}$.
By the definition of $g'_j$, we know that $d_P(j,g'_j+1)\geq \lambda_p$.
Since $\lambda_p>\lambda^*$, $d_{P}(j,g'_j+1)>\lambda^*$. Hence, it
must be that $\lambda^*\geq d^1_{C(i,a_i)}(j,g'_j+1)$.

This proves that $\lambda^*=d^1_{\max}(i,a_i)$.
The lemma thus follows.
\end{proof}

In light of Lemma~\ref{lem:80}, in the case of
$\lambda_c<\lambda_p$, $\lambda^*=\lambda_c$ is the smallest feasible
value of  $d^1_{\max}(i,a_i)$ for all $i\in \calI'$. Note that the number
of such values  $d^1_{\max}(i,a_i)$  is $O(n)$. Hence, if we can
compute $d^1_{\max}(i,a_i)$ for all $i\in \calI'$, then $\lambda^*$ can be
easily found in additional $O(n\log n)$ time using our decision
algorithm, e.g., by first sorting these values and then doing binary search.

The next lemma gives an algorithm that can compute $d^1_{\max}(i,a_i)$
for all $i\in \calI'$ in $O(n)$ time, with the help of the
range-minima data structure \cite{ref:BenderTh00,ref:HarelFa84}.

\begin{lemma}\label{lem:100}
$d^1_{\max}(i,a_i)$ for all $i\in \calI'$ can be computed in $O(n)$ time.
\end{lemma}
\begin{proof}
Consider any $i\in \calI'$. For any $j\in [i,h'_{a_i-1}]$, it is easy to
see that $d^1_{C(i,a_i)}(j,g_j+1)=|C(i,a_i)|-d_P(j,g_j+1)$, where
$|C(i,a_i)|$ is the length of the cycle $C(i,a_i)$. Hence,
we can obtain the following,
\begin{equation*}
\begin{split}
d^1_{\max}(i,a_i)  =\max_{j\in [i,h'_{a_i}-1]}d^1_{C(i,j)}(j,g'_j+1)
 & =\max_{j\in [i,h'_{a_i}-1]}\{|C(i,a_i)|-d_P(j,g'_j+1)\}\\
 & =|C(i,a_i)|-\min_{j\in [i,h'_{a_i}-1]}d_P(j,g'_j+1).
\end{split}
\end{equation*}

Define $d_{\min}(i,a_i)= \min_{j\in [i,h'_{a_i}-1]}d_P(j,g'_j+1)$.
By the above discussions we have $d^1_{\max}(i,a_i)= |C(i,a_i)|- d_{\min}(i,a_i)$.
Therefore, computing $d^1_{\max}(i,a_i)$ boils down to computing
$d_{\min}(i,a_i)$. In the following, we compute $d_{\min}(i,a_i)$ for
all $i\in \calI'$ in $O(n)$ time, after which  $d^1_{\max}(i,a_i)$ for
all $i\in \calI'$ can be computed in additional $O(n)$ time.

First of all, we compute $g'_j$ and $h'_j$ for all $j=1,2,\ldots,n$.
This can be easily done in $O(n)$ time due to the monotonicity
properties: $g'_1\leq g'_2\leq \cdots\leq g'_n$ and $h'_1\leq h'_2\leq
\cdots \leq h'_n$. Recall that for each $i\in \calI$, $a_i$ has
already been computed.
Then, we can compute $\calI'$ in $O(n)$ time by checking
whether $i\leq h'_{a_i}-1$ for each $i\in \calI$.

Next we compute an array $B[1\cdots n]$ such that
$B[j]=d_P(j,g'_j+1)$ for each $j\in [1,n]$.
Clearly, the array $B$ can be computed in $O(n)$
time. Then, we build a range-minima data structure on $B$
\cite{ref:BenderTh00,ref:HarelFa84}. The range-minima
data structure can be built in $O(n)$ time such that given any pair
$(i,j)$ with $1\leq i\leq j\leq n$,
the minimum value of the subarray $B[i\cdots j]$ can be
computed in constant time.

Finally, for each $i\in \calI'$, we can compute $d_{\min}(i,a_i)$ in
constant time by querying the range-minima data structure on $B$ with
$(i,h'_{a_i}-1)$.

Therefore, we can compute $d_{\min}(i,a_i)$ for all $i\in \calI'$, and thus
compute $d^1_{\max}(i,a_i)$ for all $i\in \calI'$ in $O(n)$ time.
\end{proof}

In summary, we can compute $\lambda^*$ in $O(n\log n)$ time in the case
$\lambda_{\gamma}<\lambda_1$ and $\lambda_c<\lambda_p$.

Our overall algorithm for computing an optimal solution is summarized
in the proof of Theorem~\ref{theo:optimization}.

\begin{theorem}\label{theo:optimization}
An optimal solution for the optimization problem can be found in $O(n\log n)$ time.
\end{theorem}
\begin{proof}
First, we compute $\lambda_{\alpha}$, $\lambda_{\beta}$, and
$\lambda_{\delta}$, in $O(n\log n)$ time by using our decision
algorithm and the sorted-matrix searching techniques. Then, we
compute
$\lambda_1=\min\{ \lambda_{\alpha}$, $\lambda_{\beta}, \lambda_{\delta} \}$.

Second, by using $\lambda_1$, we compute the indices $I_i'(\alpha)$,
$I_i'(\beta)$, and $I_i'(\delta)$ for all $i=1,2,\ldots,n$. This can
be done in $O(n)$ time. For each $i\in [1,n]$, if $[1,I'_{i}(\alpha)]\cap
[I'_{i}(\beta),n]\cap [I'_{i}(\delta),n]\neq \emptyset$, we
compute $a_i$ (i.e., the smallest index in the above intersection)
and add $i$ to the set $\calI$ (initially $\calI=\emptyset$). Hence, all
such $a_i$'s and $\calI$ can be computed in $O(n)$ time.

If $\calI=\emptyset$, then we return $\lambda_1$ as $\lambda^*$.

If $\calI\neq \emptyset$, then we compute $\lambda_p$ in $O(n\log n)$ time by
Lemma~\ref{lem:70}. We proceed to compute $d^1_{\max}(i,a_i)$ for all
$i\in \calI'$ by Lemma~\ref{lem:100}, and then find
the smallest feasible value $\lambda'$ in the set $\{d^1_{\max}(i,a_i) \ |\ i\in \calI'\}$ in $O(n\log n)$
time.
Finally, we return $\min\{\lambda_1,\lambda_p,\lambda'\}$ as
$\lambda^*$.

The above computes $\lambda^*$ in $O(n\log n)$ time. Applying
$\lambda=\lambda^*$ on our decision algorithm can eventually find an optimal
edge in additional $O(n)$ time.
\end{proof}




\bibliographystyle{plain}
\bibliography{reference}

%




\end{document}